\documentclass[submission,copyright,creativecommons]{eptcs}
\providecommand{\event}{Infinity 2012}

\usepackage{times}
\usepackage[cmex10]{amsmath}
\usepackage{amsfonts}
\usepackage{amssymb}
\usepackage{stmaryrd}
\usepackage{array}
\usepackage{latexsym}
\usepackage{verbatim}
\usepackage[english]{babel}
\usepackage{epsfig}
\usepackage{graphicx}
\usepackage{moreverb}
\usepackage{times}
\usepackage{epic,eepic}
\usepackage{subfigure}
\usepackage{float}
\usepackage[all,2cell,ps]{xy}
\usepackage{psfrag}
\usepackage{multirow}

\newcommand{\edge}[1]{\ar@{-}[#1]}

\newtheorem{definition}{Definition}
\newtheorem{theorem}{Theorem}
\newtheorem{example}{Example}
\newtheorem{proof}{Proof}

 \newcommand{\M}{\mathcal{M}}

\newcommand{\Def}{\overset{\operatorname{def}}{=}}
\newcommand{\parties}[1]{\mathcal{P}(#1)}

\newcommand{\Z}{Z\hspace{-.6em}Z}
\newcommand{\R}{I\hspace{-.2em}R}
\newcommand{\bbm}{\begin{bmatrix}}
\newcommand{\ebm}{\end{bmatrix}}

\newcommand{\sh}[1]{\ensuremath{#1^\sharp}}
\newcommand{\lfp}{\ensuremath{\textrm{lfp}}}
\newcommand{\gfp}{\ensuremath{\textrm{gfp}}}
\newcommand{\treillisc}{\ensuremath{
(\mathcal{D}, \sqsubseteq, \bigsqcup, \bigsqcap)}}

\title{Constraint-based reachability}
\author{Arnaud Gotlieb
\institute{Certus Software V\&V Center\\
SIMULA Research Laboratory\\
Lysaker, Norway}
\email{\quad arnaud@simula.no}
\and
Tristan Denmat
\institute{INRIA Rennes Bretagne-Atlantique\\
Rennes, France}
\email{\quad Tristan.Denmat@inria.fr}
\and
Nadjib Lazaar
\institute{LIRMM\\
Montpellier, France}
\email{\quad lazaar@lirmm.fr}
}
\def\titlerunning{Constraint-based reachability}
\def\authorrunning{A. Gotlieb, T. Denmat, N. Lazaar}

\begin{document}
\maketitle

\begin{abstract}
Iterative imperative programs can be considered as infinite-state systems computing over possibly unbounded domains. Studying reachability in these systems is challenging as it requires to deal with an infinite number of states with standard backward or forward exploration strategies. An approach that we call {\it Constraint-based reachability}, is proposed to address reachability problems by exploring program states using a constraint model of the whole program. The keypoint of the approach is to interpret imperative constructions such as conditionals, loops, array and memory manipulations with the fundamental notion of {\it constraint} over a computational domain. By combining constraint filtering and abstraction techniques, {\it Constraint-based reachability} is able to solve reachability problems which are usually outside the scope of backward or forward exploration strategies. This paper proposes an interpretation of classical filtering consistencies used in Constraint Programming as abstract domain computations, and shows how this approach can be used to produce a constraint solver that efficiently generates solutions for reachability problems that are unsolvable by other approaches.
\end{abstract}

\section{Introduction}
Modern automated program verification can be seen as the convergence of three distinct approaches, namely Software Testing, Model-Checking and Program Proving. Even if the general verification problems are often undecidable, investigations on these approaches have delivered the most efficient automated techniques to show that a given property is satisifed or not by all the reachable states of an infinite-state system. 

Several authors have advocated the usage of {\it constraints} to represent an infinite set of states and the usage of constraint solvers to efficiently address reachability problems \cite{Car96,DP01,Fla04,BCE08}. In automated program verification problems, the goal is to find a state of the program which violates a given safety property, i.e., an {\it unsafe state}. Two distinct strategies have been investigated to explore programs with constraints, namely the forward analysis and the backward analysis strategies. In forward analysis, a set of reachable states is explored by computing the transition from the initial states of a program to the next states in forward way. If an unsafe state is detected to belong to the set of reachable states during this exploration then a property violation is reported. In backward analysis, states are computed from an hypotetical unsafe state in a backward way with the hope to discover that one of those is actually an initial state. One advantage of backward analysis over forward analysis is its usage of the targeted unsafe state to refine the state search space.  
However, both strategies are quite powerful and have been implemented into several software model checkers based on constraint solving \cite{HJM03,Fla04} and automated test case generators \cite{WMM05,GKS05,FW07,CG10,BH11}.
 
In this paper, we present an integrated constraint-based strategy that can benefit from the strengths of both forward and backward analysis. The keypoint of the approach, that we have called {\it Constraint-Based Reachability (CBR)}, is to interpret imperative constructions such as conditionals, loops, array and memory manipulations with the fundamental notion of {\it constraint} over a computational domain. By combining constraint filtering and abstraction techniques, CBR is able to solve reachability problems which are usually outside the scope of backward or forward exploration strategies. A main difference is that CBR does not sequentially explore the execution paths of the program~; the exploration is driven by the constraint solver which picks-up the constraint to explore depending on the priorities that are attached to them. It is worth noticing that applying CBR to program exploration results in a semi-correct procedure only, meaning that there is no termination guarantee. CBR has been mainly applied in automatic test data generation for iterative programs \cite{GBR98,GBR00}, programs that manipulate pointers towards named locations of the memory \cite{GDB05b,GDB07}, programs on dynamic data structures and anonymous locations \cite{CBG09}, programs containing floating-point computations \cite{BGM06}. A major improvement of the approach was brought by the usage of Abstract Interpretation techniques to enrich the filtering capabilities of the constraints used to represent conditionals and loops \cite{DGD07a,DGD07b}. This approach permitted us to build efficient test data generator tools for a subset of C \cite{Got09} and Java Bytecode \cite{CG10}.

The first contribution of this paper is the interpretation of classical filtering consistencies notions in terms of 
abstract domain computations. Constraint filtering is the main approach behind the processing of constraints in a finite domains constraint solver. We show in general the existence of tight links between classical filtering techniques and abstract domain computations that were not pointed out elsewhere. We also give the definition of a new consistency filtering inspired from the Polyhedral abstract domain, as consequence of these links. 

The second contribution is the description of a special constraint handling any iterative construction. The constraint {\it w} captures iterative reasoning in a constraint solver and as such, is able to deduce information which is outside the scope of any pure forward or backward abstract analyzer. Its filtering capabilities combines both constraint reasoning and abstract domain computations in order to propagate informations to the rest of the constraint system. In this paper, we focus on the theoretical foundations of the constraints, while giving examples of its usage for test case generation over iterative programs.

{\bf Outline of the paper.}
The rest of the paper is organized as follows. Sec.2 introduces the necessary background in Abstract Interpretation to understand the contributions of the paper. Sec.3 establishes the link between classical constraint filtering and abstract domain computations. Sec.4 describes the theoretical foundation of the {\it w} constraint for handling iterative constructions while Sec.5 concludes the paper.

\section{Background}


Abstract Interpretation (AI) is a theoretical framework introduced by Cousot and Cousot in \cite{CC77} to manipulate
abstractions of program states. An abstraction can be used to simplify program analysis problems otherwise not
computable in realistic time, to manageable problems more easily solvable.
Instead of working on the concrete semantics of a 
program\footnote{Program semantics captures formally all the possible behaviours of a program.}, 
AI computes results over an abstract semantics 
allowing so to produce over-approximating properties of the concrete semantics.
In the following we introduce the basic notions required to understand AI. 

\begin{definition}[Partially ordered set (poset)]
  Let $\sqsubseteq$ be a partial order law, then
  the pair $(\mathcal{D},\sqsubseteq)$ is called a {\it poset} iff
\begin{displaymath}
\begin{array}{ll}
\forall x \in \mathcal{D}, x \sqsubseteq x & \textrm{\textit{(reflexive)}} \\
\forall x,y \in \mathcal{D}, x \sqsubseteq y \land y \sqsubseteq x \implies x = y & \textrm{\textit{(anti-symmetry)}} \\
\forall x,y,z \in \mathcal{D}, x \sqsubseteq y \land y \sqsubseteq z \implies x \sqsubseteq z & \textrm{\textit{(transitive)}} \\
\end{array}
\end{displaymath}
\end{definition}

\begin{definition}[Complete lattice]
  A complete lattice is a 4-tuple $(\mathcal{D}, \sqsubseteq, \bigsqcup, \bigsqcap)$ such that
\begin{itemize}
\item $(\mathcal{D},\sqsubseteq)$ is a poset 
\item $\bigsqcup$ is a upper bound: $\forall \mathcal{S} \subseteq \mathcal{D}$, we have
  \begin{eqnarray*}
    && \forall x \in \mathcal{S}, x \sqsubseteq \bigsqcup \mathcal{S}\\
    && \forall y \in \mathcal{D},(\forall x \in \mathcal{S}, x \sqsubseteq y) \implies \bigsqcup \mathcal{S} \sqsubseteq y
  \end{eqnarray*}
\item   $\bigsqcap$ is a lower bound: $\forall \mathcal{S} \subseteq \mathcal{D}$, we have
  \begin{eqnarray*}
    && \forall x \in \mathcal{S}, \bigsqcap \mathcal{S}  \sqsubseteq x\\
    && \forall y \in \mathcal{D},(\forall x \in \mathcal{S}, y \sqsubseteq x) \implies y  \sqsubseteq \bigsqcap \mathcal{S}
  \end{eqnarray*}
\end{itemize}
\end{definition}
\noindent 
Complete lattices have a single smallest element $\bot = \bigsqcap \mathcal{D}$ 
and a single greatest element $\top = \bigsqcup \mathcal{D}$. 
Program semantics can usually be expressed as the least fix point of a monotonic and continuous function.
A function $f$ from a complete lattice $\treillisc$ to itself is monotonic iff 
$\forall l_1,l_2 \in \mathcal{D}, l_1 \sqsubseteq l_2 \implies f(l_1) \sqsubseteq f(l_2)$.
It is continuous iff
$\forall \mathcal{S} \subseteq \mathcal{D}, f(\bigsqcup
\mathcal{S}) = \bigsqcup_{s \in \mathcal{S}}(f(s))$ and 
$f(\bigsqcap \mathcal{S})= \bigsqcap_{s \in S}(f(s))$. 

\noindent
The following Theorem guarantees the existence
of the fix points of a monotonic function.
\begin{theorem}[Knaster-Tarski]
  In a complete lattice $\treillisc$, for all monotonic functions\\ 
$f: \mathcal{D} \rightarrow \mathcal{D}$, 
\begin{itemize}
\item the least fix point of $f$ (i.e., $lfp(f)$) exists and $lfp(f) = \bigsqcap\{x~|~f(x) \sqsubseteq x\}$
\item the greatest fix point of $f$ (i.e., $gfp(f)$) exists and $gfp(f) = \bigsqcup\{x~|~f(x) \sqsubseteq x\}$
\end{itemize}
\end{theorem}
\noindent
In addition, when the functions are continuous, these fix points can be computed using an algorithm derived from the
following theorem:
\begin{theorem}[Kleene]\label{theo:limites}
  In a complete lattice $\treillisc$, for all monotonic and continuous functions 
  $f : \mathcal{D} \rightarrow \mathcal{D}$, the least fix point of
  $f$ is equal to $\bigsqcup\{f^n(\bot)~|~n \in \mathbb{N}\}$ and the greatest fix point
  of $f$ is equal to $\bigsqcap\{f^n(\top)~|~n \in \mathbb{N}\}$
\end{theorem}
\noindent
As $\bot, f(\bot), \ldots, f^n(\bot), \ldots$ is an increasing suite, we get
$\bigsqcup\{f^n(\bot)~|~n \leq k\} = f^k(\bot)$.
Hence, $\lfp(f) = \lim_{k \rightarrow +\infty} f^k(\bot)$
and $\gfp(f) = \lim_{k \rightarrow +\infty} f^k(\top)$.

\noindent
For reaching the least fix point of a monotonic and continuous function in a complete lattice, it
suffices to iterate $f$ from $\bot$ until a fix point is reached.

\noindent
Let $(\mathcal{D}, \sqsubseteq,\bigsqcup, \bigsqcap)$ be a complete lattice called the {\it concrete lattice} 
and $f$ a function that defines some concrete semantics over this lattice, let
$(\sh{\mathcal{D}}, \sh{\sqsubseteq})$ be a poset called the {\it abstract poset}, and
$\sh{f} : \sh{\mathcal{D}} \rightarrow \sh{\mathcal{D}}$ be a continuous function,
then {\it Abstract Interpretation} aims at computing a fix point of $\sh{f}$ in order to
over-approximate the computation performed by $f$. 

\noindent
Depending on whether the abstract poset is a complete lattice or not,
we have distinct theoretical results regarding the abstraction. Proofs of
the following theorems can be found in~\cite{CH78}.

\paragraph{Galois connection}
When the abstract poset is a complete lattice, the notion of {\it Galois connection} is available
to link the abstract computations with the concrete lattice. 

\begin{definition}[Galois connection]
Let $(\mathcal{D}, \sqsubseteq, \bigsqcup, \bigsqcap)$  
and $(\sh{\mathcal{D}}, \sh{\sqsubseteq}, \sh{\bigsqcup}, \sh{\bigsqcap})$ be two complete lattices,
then a pair of functions
$\alpha : \mathcal{D} \rightarrow \sh{\mathcal{D}}$ and
$\gamma : \sh{\mathcal{D}} \rightarrow  \mathcal{D}$ is a {\it Galois connection} iff
$\forall x \in \mathcal{D}, \forall y \in \sh{\mathcal{D}}, \alpha(x) \sh{\sqsubseteq} y \iff x \sqsubseteq \gamma(y)$
noted:
\begin{center}
$(\mathcal{D}, \sqsubseteq, \bigsqcup, \bigsqcap) \overset{\gamma}{\underset{\alpha}{\leftrightarrows}} 
(\mathcal{D}^\sharp, \sqsubseteq^\sharp, \bigsqcup^\sharp, \sh{\bigsqcap})
$
\end{center}
\end{definition}
\noindent
Next definition establishes the correction property of an analysis.
\begin{definition}[Sound approximation]\label{correct-app}
Let 
$(\mathcal{D}, \sqsubseteq, \bigsqcup, \bigsqcap) \overset{\gamma}{\underset{\alpha}{\leftrightarrows}} (\mathcal{D}^\sharp, \sqsubseteq^\sharp, \bigsqcup^\sharp, \sh{\bigsqcap}) $ be a Galois connection, then a function
$\sh{f} : \sh{\mathcal{D}} \rightarrow \sh{\mathcal{D}}$
is a {\it sound approximation} of $f : \mathcal{D} \rightarrow \mathcal{D}$ iff
$$ \forall y \in \mathcal{D}^{\sharp}, f \circ \gamma (y) \sqsubseteq \gamma \circ f^\sharp(y) $$
\end{definition}
\noindent
Consquently, we have the following notion:
\begin{theorem} [Smallest sound approximation]
 \label{theo:approx}
Let $(\mathcal{D}, \sqsubseteq, \bigsqcup, \bigsqcap) \overset{\gamma}{\underset{\alpha}{\leftrightarrows}} (\mathcal{D}^\sharp, \sqsubseteq^\sharp, \bigsqcup^\sharp,\sh{\bigsqcap}) $ be a Galois connection,
and a function $f : \mathcal{D} \rightarrow \mathcal{D}$, then the {\it smallest sound approximation} of $f$
is $\alpha \circ f \circ \gamma$
\end{theorem}
\noindent
This theorem implies that any function greater than
$\alpha \circ f \circ \gamma$ is a sound approximation of $f$ and
the following theorem characterizes the results of fixpoint computations:
\begin{theorem}[Fixpoint computations with sound approximation]
  Let $(\mathcal{D}, \sqsubseteq,\bigsqcup, \bigsqcap)
  \overset{\gamma}{\underset{\alpha}{\leftrightarrows}}
  (\mathcal{D}^\sharp, \sqsubseteq^\sharp, \bigsqcup^\sharp, \sh{\bigsqcap}) $ be
a Galois connection, let $\sh{f}:
  \sh{\mathcal{D}} \rightarrow \sh{\mathcal{D}}$ and $f: \mathcal{D}\rightarrow \mathcal{D}$ 
be two monotonic functions such that $\sh{f}$ is a sound approximation of $f$, 
then, we have:
\begin{eqnarray*}
&&\lfp(f) \sqsubseteq \gamma(\lfp(\sh{f}))~\textrm{and} \\
&&\gfp(f) \sqsubseteq \gamma(\gfp(\sh{f}))
\end{eqnarray*}
\end{theorem}
\noindent
Intuitively, this theorem gives a process to compute an over-approximation by
Abstract Interpretation, as shown in Fig~\ref{fig:algo_ia}.
\begin{figure}
\begin{center}
  \includegraphics[width=10cm,height=7cm]{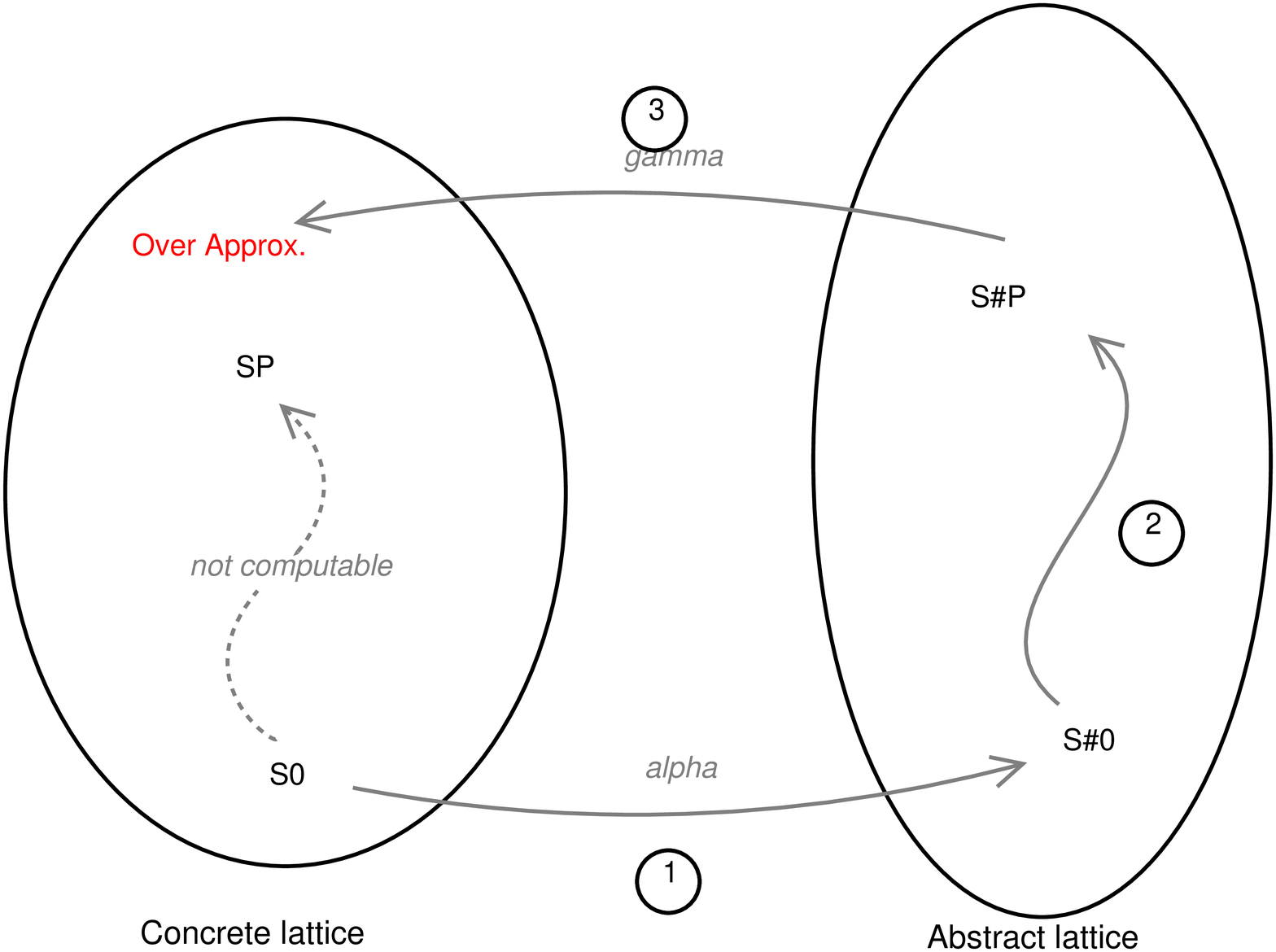}
\end{center}
\caption{Static approximations of fixed point computations in complete lattices}
\label{fig:algo_ia}
\end{figure}
The left part shows the concrete lattice where the concrete computation of $f$
is performed starting from initial state $S_0$. The right part shows the
abstract lattice that is used to over-approximate the computation.
This computation is undertaken in three steps:
\begin{itemize}
\item initial state abstraction;
\item fixpoint computation in the abstract lattice;
\item result concretization.
\end{itemize}

\paragraph{Without Galois connection}

When the abstract lattice is not complete, there does not exist necessarily a best
abstraction for all elements of the concrete lattice. The notion of Galois connection is
no more available and the abstract lattice is just linked with the concrete lattice
through a monotonic function $\gamma: \sh{\mathcal{D}} \rightarrow \mathcal{D}$.
The definition of sound approximation needs to be adapted:
\begin{definition} [Sound approximation without a Galois connection]
  Let $(\mathcal{D}, \sqsubseteq)$ and
  $(\mathcal{D}^\sharp, \sqsubseteq^\sharp)$ be two posets, let 
  $\gamma : \mathcal{D}^{\sharp} \rightarrow \mathcal{D}$ be a monotonic function and 
  $f : \mathcal{D} \rightarrow \mathcal{D}$ a function, then the function
  $\sh{f} : \sh{\mathcal{D}} \rightarrow \sh{\mathcal{D}}$ is a sound approximation of $f$ iff
$$
\forall x \in \mathcal{D}^{\sharp}, f \circ \gamma (x) \sqsubseteq \gamma \circ f^\sharp(x)
$$
\end{definition}
\noindent
In such an (not complete) abstract lattice, nothing guarantees the existence of the least fix point:
$\lfp(f)$ is not necessarily approximated by $\lfp(\sh{f})$. However, any fix point of $\sh{f}$ can be
used:
\begin{theorem}
  Let $(\mathcal{D}, \sqsubseteq, \bigsqcup,
  \bigsqcap)$ be a complete lattice, and $(\mathcal{D}^\sharp, \sqsubseteq^\sharp)$ be a poset, let
  $\gamma :\sh{\mathcal{D}} \rightarrow
    \mathcal{D}$, $f : \mathcal{D} \rightarrow \mathcal{D}$ and $\sh{f}
    : \sh{\mathcal{D}} \rightarrow \sh{\mathcal{D}}$ be three monotonic functions then
if $\sh{f}$ is a sound approximation of $f$, then we have:
\begin{eqnarray*}
\forall x \in \mathcal{D}^\sharp, \sh{f}(x) = x \implies \lfp(f) \sqsubseteq \gamma(x)
\end{eqnarray*}
\end{theorem}
\noindent
Next theorem is useful to compute an over-approximation of $\gfp(f)$ when the lattice is not complete:
\begin{theorem}
\label{theo:gfp_incomplet}
  Let $(\mathcal{D}, \sqsubseteq, \bigsqcup,
  \bigsqcap)$ be a complete lattice, let $(\mathcal{D}^\sharp, \sqsubseteq^\sharp)$ be a poset with
  a greatest element $\top$ and let
  $\gamma : \sh{\mathcal{D}} \rightarrow \mathcal{D}$, $f :
  \mathcal{D} \rightarrow \mathcal{D}$ and $\sh{f} : \sh{\mathcal{D}}
  \rightarrow \sh{\mathcal{D}}$ be three monotonic functions, then
  \\
  if $\sh{f}$ is a sound approximation of $f$ and $a$ is an
  element of $\sh{\mathcal{D}}$ such that there exists $k$ such as $a =
  {\sh{f}}^k(\top)$, then
$$
\gfp(f) \sqsubseteq \gamma(a)
$$
\end{theorem}
\noindent
Consequently, when the abstract lattice is not complete, instead of
abstracting the initial state, one selects
an element of the abstract lattice that over-approximates the initial state.
And, a fix point is computed in the abstract lattice from this element.
The fix point is still an over-approximation of the concrete semantics.

\subsection{Examples of abstract domains}
\label{sec:ia:exemples}

In this section, we briefly describe two abstract domains:
the Interval \cite{CC76} and the Polyhedral \cite{CH78} domains. 

\subsubsection{The Interval abstract domain}
\label{intervals}
Interval analysis aims at approximating a set of values by an interval of possible values.
If $\mathcal{I}_{\mathbb{N}} = \{[a,b]~|~a,b \in \mathbb{N} \cup
\{-\infty, +\infty\}\}$, then the Interval abstract domain is the Cartesian product 
$\mathcal{I}_{\mathbb{N}} \times ... \times \mathcal{I}_{\mathbb{N}}$ equipped with
inclusion, union and intersection over intervals. This abstract domain is a complete lattice.

\noindent
State abstraction is performed by computing an interval that over-approximates the set of possible
values for each variable. If the concrete state is an unbounded set of tuples
$\{(x_{11},\ldots,x_{n1}),(x_{12},\ldots,x_{n2}),\ldots\}$ then:
\begin{multline*}
\alpha(\{(x_{11},\ldots,x_{n1}),(x_{12},\ldots,x_{n2}),\ldots\}) =
([m_1,M_1],\ldots,[m_n,M_n]) \\ ~\textrm{with}~m_i= \left\{ \begin{array}{ll} min_j(x_{ij})
  ~\textrm{if it exists} \\ -\infty~\textrm{else} \end{array} \right.~\textrm{and}~M_i = \left\{ \begin{array}{ll} max_j(x_{ij})
  ~\textrm{if it exists} \\ +\infty~\textrm{else} \end{array} \right.
\end{multline*}
The concretization of an abstract state is obtained by computing the Cartesian product of the intervals.
These functions define a Galois connection between the concrete domain and the abstract domain of intervals.

\noindent
The approximation of transfert functions is realized by using their structure and classical results from Interval Analysis \cite{Moo66}.
For example, functions $[x \leq y]$ and $[x=y+z]$ are abstracted by the following (sound) approximations:
$\sh{[x \leq y]}: ([a,b],[c,d]) \rightarrow ([a,min(b,d)],[max(a,c),d]))$ and 
$\sh{[x=y+z]}: ([a,b], [c,d], [e,f]) \rightarrow ([c+e,d+f] \sh{\cap} [a,b], [a-f,e-b] \sh{\cap} [c,d],[a-d, c-b] \sh{\cap} [e,f])$.    

\subsubsection{The Polyhedral abstract domain}
\label{polyhedra}

In Polyhedral analyses, each concrete state is abstracted by a conjunction of linear
constraints that defines a convex polyhedron. Indeed, a {\it convex polyhedron} is a region of an
n-dimensional space that is bounded by a finite set of hyperplanes ${x \in \R^n | ax \geq c}$ where $a \in \R^n$ and $c \in \R$.
 The abstract lattice equiped with
inclusion, convex hull\footnote{The union of two polyhedra is not a polyhedron, this is the reason why convex hull or any relaxation of it
must be employed.}, and intersection of polyhedra is not a complete lattice as there is no 
upper bound to the convex union of all the convex polyhedra that can be written in a circle.    

\noindent
Abstract functions can be defined to deal with polyhedra. For example:
\begin{eqnarray}
\sh{[x \geq y]}(\{z \leq  x + y\}) &=&  \{z \leq x + y, y \leq x\} \label{eq:poly_1}\\
\sh{[x > y ]}(\{x \leq y \}) &=& \{0 = 1\} \label{eq:poly_2}\\
\sh{[x = y * z]}(\{1 \leq y \leq 10\}) &=& \{x \leq z, x \leq 10*z\} \label{eq:poly_3}
\end{eqnarray}
\noindent 
If the expression is a linear condition, then it is just added to the polyhedron (case~\ref{eq:poly_1}). 
If the expression is contradictory with the current polyhedron, then it is reduced to $1=0$ meaning that
there is no abstract (and concrete) state in the approximation (case~\ref{eq:poly_2}). 
If the expression is non-linear, then a linear approximation is derived when available
and added to the polyhedron (case~\ref{eq:poly_3}). 

\section{Filtering consistencies as abstract domain computations}
\label{sec:absdom-cons}
As noticed by Apt \cite{Apt99a}, constraint propagation algorithms can be seen as
instances of algorithms that deal with chaotic iteration. In this context, chaotic means fair application of
propagators until saturation. In this section, we elaborate on a bridge between two
unrelated notions: filtering consistencies and abstract domains. In particular, we show that arc-- and 
bound-- consistency are instances of chaotic iterations over two distinct 
abstract domains. Classical AI notions of sound approximation and abstract domain
computations, not used in \cite{Apt99a}, allows to show that filtering consistencies compute sound 
over-approximations of the solutions set of a constraint system. Thanks to the bridge, we also propose new filtering
consistency algorithms based on the polyhedral abstract domain. 

\subsection{Notations}
Let $\Z$ be the set of integers and $\cal V$ be a finite set of integer variables, where each variable $x$ in $\cal V$ is
associated with a finite domain $D(x)$. The {\it domain} 
$\mathcal{D}$ is the Cartesian product of each variable domain:~$D(x_1) \times \ldots \times D(x_m)$ and
$\mathcal{P(D)}$ denotes the powerset of $\mathcal{D}$. $inf_{\mathcal{D}} x$ and $sup_{\mathcal{D}} x$ denote
respectivelly the inferior and the superior bounds of $D(x)$ in $\mathcal{D}$.
A constraint $c$ is
a relation between variables of $\cal V$. The language of (elementary) constraints is built over arithmetical
operators $\{+,-,*,\,...\}$ and relational operators $\{<,\leq,>,\geq, =, \neq, ...\}$ but any relation over a subset
of $\cal V$ can be considered. Let {\it vars}$(c)$ be the function that returns the variables of $\cal V$ 
appearing in a constraint $c$. A {\it valuation} $\sigma$
is a mapping of variables to values, noted $\{x_1 \mapsto d_1, ...,x_n \mapsto d_n \}$. $CS$ denotes a constraint system $CS$, i.e., 
a finite set of constraints.

\subsection{Exact filtering}

Let $\{c_1,..,c_m\}$ be a $CS$ over $\{x_1,..,x_n\}$ and let $\mathcal{D} = D(x_1) \times .. \times D(x_n)$, then
the solution-set of $CS$ is an element of $\mathcal{P(D)}$, noted $sol(CS)$.

\noindent
The {\it exact filtering operator} of a constraint $c_i$ is computed with the function
$f_i : \parties{\mathcal{D}} \rightarrow \parties{\mathcal{D}}$ which maps an element $S \in \parties{\mathcal{D}}$ to 
$f_i(S) = \{s~|~ s\in S \land c_i(s)\}$.
The exact filtering operator of $c_i$ removes {\it all} the tuples of $\mathcal{D}$ that violate $c_i$.
Hence, by using an iterating procedure, it permits to compute $sol(CS)$:
if $f_C = f_1 \circ \ldots \circ f_m$ then $sol(C) = \gfp(f_C)$.
By noticing that $f_C$ is continuous (as each $f_i$ is continuous) and monotonic and thanks to Theorem~\ref{theo:limites} 
we get $sol(CS) = \lim_{k \rightarrow +\infty} {f_C}^k(\mathcal{D})$. 
\begin{example}
Consider $CS=\{x \neq y, y \neq z, z \neq x\}$ where $x \in 1..2, y \in 1..2, z \in 1..2$. The exact
filtering operator associated with $x \neq y$ will remove the tuples $(1,1,1),(1,1,2),(2,2,1),(2,2,2)$ from 
$\{1,2\} \times \{1,2\} \times \{1,2\}$. Iterating over all the constraints of $CS$ will eventually exhibit the inconsistency of this example.
\end{example}
\noindent
In fact, this shows that exact filtering of a CS over $D(x_1) \times .. \times D(x_n)$ 
can be reached if one computes over a complete lattice built over the set of possible valuations: 
$(\mathcal{P}(D(x_1) \times .. \times D(x_n)), \subseteq, \bigcup,\bigcap)$. 
This lattice will be called the {\it concrete lattice} in the rest of the paper.
Of course, computing over the concrete lattice is usually unreasonable, as it requires to examine every tuple of the Cartesian product 
$D(x_1) \times .. \times D(x_n)$ w.r.t. consistency of each constraint.
 
\subsection{Domain-consistency filtering}

For binary constraint systems, the most successful local consistency filtering is arc-consistency, 
which ensures that every value in the domain of one variable has a support in the domain of the other variable.
The standard extension of arc-consistency for constraints of more than two variables is domain-consistency (also called
hyper-arc consistency \cite{MS98}).
Roughly speaking, the abstraction that underpins domain-consistency filtering aims at considering
each variable domain separately, instead of considering the Cartesian product of each individual domain. More
formally,
\begin{definition}[Domain-consistency]
A domain $\mathcal{D}$ is domain-consistent for a constraint $c$ where $vars(c)=\{x_1,..,x_n\}$ 
iff for each variable $x_i$, $1 \leq i \leq n$ and for each $d_i \in D(x_i)$
there exist integers $d_j$ with $d_j \in D(x_j)$, $1 \leq j \leq n, j \neq i$
such that $\sigma = \{x_1 \mapsto d_1, ..,x_n \mapsto d_n\}$ is an integer solution of $c$.
\end{definition} 
\noindent
Consider the domains $\mathcal{D}=D(x_1) \times .. \times D(x_n)$ and  
$\sh{\mathcal{D}}_{arc} = \mathcal{P}(D(x_1)) \times \ldots \times \mathcal{P}(D(x_n))$ and
the abstraction function $\alpha_{arc} : \parties{\mathcal{D}} \rightarrow
\sh{\mathcal{D}}_{arc}$ which maps $S \in \parties{\mathcal{D}}$ to
$$ \alpha_{arc}(S) = (\{x_1~|~x \in S\}, \ldots, \{x_n~|~ x \in S\})$$

\noindent
The concretization function is a function $\gamma_{arc} : \sh{\mathcal{D}}_{arc} \rightarrow \parties{\mathcal{D}}$ such that
$$\gamma_{arc}((S_1,\ldots,S_n)) = S_1 \times \ldots \times S_n $$
\noindent
If $\sh{\sqsubseteq}_{arc}$, $\sh{\bigsqcup}_{arc}$ and  $\sh{\bigsqcap}_{arc}$ denote respectivelly the
inclusion, union and intersection of two tuples of sets, then we got the following Galois connection:
\begin{center}
$
(\parties{\mathcal{D}}, \subseteq, \bigcup, \bigcap) \overset{\gamma_{arc}}{\underset{\alpha_{arc}}{\leftrightarrows}} 
(\sh{\mathcal{D}}_{arc}, \sh{\sqsubseteq}_{arc}, \sh{\bigsqcup}_{arc}, \sh{\bigsqcap}_{arc})
$
\end{center}
\noindent
The proof follows comes the monotonicity of the projection and Cartesian product.
From Theorem~\ref{theo:approx}, we get:
\begin{definition}
The best sound approximation of the exact filtering operator $f_{i}$ is 
$$ \sh{f_{i\_arc}} \Def \alpha_{arc}\circ f_i \circ \gamma_{arc} $$
\end{definition}

\begin{theorem}
\label{theo:arc}
  Let $p$ be a filtering operator associated with constraint $c_i$, then $p$ computes domain-consistency iff
  $p = \sh{f_{i\_arc}}$. 
\end{theorem}
\noindent
This theorem implies that domain-consistency is the strongest property that can be guaranteed by a filtering 
operator using the abstraction $\alpha_{arc}$. A proof is given in the Appendix of the paper.

\noindent
Let us consider now the function $\sh{f}_{arc}$ such that 
$\sh{f}_{arc} = \sh{f_{1\_arc}} \circ \ldots \circ \sh{f_{n\_arc}}$.
As $\sh{f}_{arc}$ is a sound approximation of $f_C$ then
$$ sol(C) = \gfp(f_C)  \subseteq \gamma_{arc}(\gfp(\sh{f}_{arc})) $$

\noindent
This result shows if necessary that constraint propagation over domain-consistency filtering 
operators computes an over-approximation of the solution set of $C$.

\subsection{Bound-consistency filtering}
\label{sec:bornes}

Following the same scheme, AI can be used to show the abstraction that underpins
constraint propagation with bound-consistency filtering (also called interval-consistency). But, firstly, let
us recall the definition of bound-consistency we consider in this paper, as several definitions exist in the 
literature \cite{CHL06}~:
\begin{definition}[Bound-consistency]
A domain $\mathcal{D}$ is bound-consistent for a constraint $c$ where $vars(c)=\{x_1,..,x_n\}$ 
iff for each variable $x_i$, $1 \leq i \leq n$ and for each $d_i \in \{inf_{\mathcal{D}} x_i, sup_{\mathcal{D}} x_i\}$
there exist integers $d_j$ with $inf_{\mathcal{D}} x_j\leq d_j \leq sup_{\mathcal{D}} x_j$, $1 \leq j \leq n, j \neq i$
such that $\sigma = \{x_1 \mapsto d_1, ..,x_n \mapsto d_n\}$ is an integer solution of $c$.
\end{definition}
\noindent 
Roughly speaking, this approximation considers only the bounds of the domain of each variable and
approximates each domain with an interval. 
Let $\mathcal{I}(S)= [min(S),max(S)]$ be the smallest interval that contains all the elements of a finite set of integers $S$.
Similarly, $\mathcal{I}^{-1}(I)$ denotes the set of integers of an interval $I$~: $\mathcal{I}^{-1}([a,b]) = \{x \in
\mathbb{Z}~|~a\leq x \leq b\}$. 

\noindent
The abstract domain we consider for bound-consistency is 
$\sh{\mathcal{D}}_{bound} = \mathcal{I}(\mathcal{P}(D(x_1))) \times
\ldots \times \mathcal{I}(\mathcal{P}(D(x_n)))$.

\noindent
Given a tuple of sets $(S_1,\ldots,S_n)$ and a tuple of intervals
$(I_1, \ldots, I_n)$, we consider the functions 
$\alpha_{inter}$ and $\gamma_{inter}$ such that:
\begin{eqnarray*}
&&\alpha_{inter}(S_1,\ldots,S_n) = (\mathcal{I}(S_1), \ldots, \mathcal{I}(S_n)) \\
&&\gamma_{inter}(I_1, \ldots I_n) = (\mathcal{I}^{-1}(I_1), \ldots, \mathcal{I}^{-1}(I_n))
\end{eqnarray*}

\noindent
Let $\alpha_{bound} : \parties{\mathcal{D}} \rightarrow \sh{\mathcal{D}}_{bound}$ be an abstraction function such that 
$$ \alpha_{bound} = \alpha_{inter} \circ \alpha_{arc}$$
and $\gamma_{bound} : \sh{\mathcal{D}}_{bound} \rightarrow
\parties{\mathcal{D}}$ be a concretization function such that 
$$ \gamma_{bound} = \gamma_{arc} \circ \gamma_{inter}$$

\noindent
If $\sh{\sqsubseteq}_{bound}$, $\sh{\bigsqcup}_{bound}$ and
$\sh{\bigsqcap}_{bound}$ respectively denote inclusion, union and
intersection of intervals (component by component) then we get the following Galois connection:

\begin{center}
$
(\parties{\mathcal{D}}, \subseteq, \bigcup, \bigcap) \overset{\gamma_{bound}}{\underset{\alpha_{bound}}{\leftrightarrows}} 
(\sh{\mathcal{D}}_{bound}, \sh{\sqsubseteq}_{bound}, \sh{\bigsqcup}_{bound}, \sh{\bigsqcap}_{bound})
$
\end{center}

\noindent
Let $\sh{f_{i\_bound}}$ be the most accurate sound approximation of $f_i$, then we get:

\begin{eqnarray*}
\sh{f_{i\_bound}} &=& \alpha_{bound} \circ f_i \circ \gamma_{bound} \\
& =& \alpha_{inter} \circ \sh{f_{i\_arc}} \circ \gamma_{inter} 
\end{eqnarray*}

\begin{theorem}
 If $p$ is a filtering operator associated to constraint $c_i$, then  $p$ computes 
 bound-consistency iff $p = \sh{f_{i\_bound}}$. 
\end{theorem}

\noindent
This theorem, proved in Appendix, implies that bound-consistency is the strongest property that can be reached with an operator
based on the $\alpha_{bound}$ abstraction.

\noindent
Consider now the function $\sh{f}_{bound}$ such that
$\sh{f}_{bound} = \sh{f_{1\_bound}} \circ \ldots \circ \sh{f_{n\_bound}}$. As
$\sh{f}_{bound}$ is a sound approximation of $f_C$, then
$$ sol(C) = \gfp(f_C)  \subseteq \gamma_{bound}(\gfp(\sh{f}_{bound})) $$
This result shows if necessary that constraint propagation based on bound-consistency computes
a sound over-approximation of the solution set of $C$.
In addition, as $\sh{f}_{bound}$ is also a sound over-approximation of $\sh{f}_{arc}$, then 
$$ \gamma_{arc}(\gfp(\sh{f}_{arc})) \subseteq \gamma_{bound}(\gfp(\sh{f}_{bound})) $$
meaning that filtering with bound-consistency provides an over-approximation of the results given by a 
filtering with domain-consistency.
 
\subsection{New filtering consistencies based on abstract domains}
\label{sec:DynaLIB_global}
In the previous section, classical filtering consistencies are interpreted in terms
of abstract domain computations. In this section, we propose a
new filtering consistency based on the Polyhedral abstract domain \cite{CH78}. 

\subsubsection{Linear relaxations}\label{sec:DLRs}

When non-linear constraints are involved in a constraint store, approximating them 
with linear constraints is natural in order to benefit from powerful Linear
Programming techniques. These techniques can be used to check the satisfiability 
of the constraint store when the approximation is sound. If the approximate constraint 
system is unsatisfiable so is the non-linear constraint system. But, in the context
of optimization problems, the approximation
can also be used to prune current bounds of the function to optimize. 

Another form of approximation comes from the domain in which the computation occurs.
A linear problem over integers can be relaxed in the domain of rationals or reals 
and solved within this domain. As the set of integers belongs to the rationals and reals,
an integer solution of the relaxed problem is also a solution of the original integer problem,
but the converse is false. In this paper, we will consider both kinds of approximations 
under the generic term of ``linear relaxations''.

Computing a linear relaxation of a constraint system $CS$ aims at finding a set of
linear constraints that characterizes an over-approximation of the solution set of
$CS$. It is not unique but for trivial reasons, we are more interested in the tighter 
possible relaxations. The tightest linear relaxation is the convex hull of the solution set
of $CS$ but computing this relaxation is as hard as solving $CS$. For $CS$ over finite domains,
the problem is therefore NP\_hard. 
Whenever a relaxation is computed by using the current bounds 
of variable domains, it is called {\it dynamic} and the consistencies presented in the rest
of the section are compatible with dynamic linear relaxations.

\subsubsection{Polyhedral-consistency filtering}

Let $Poly$ be the abstract domain of closed convex polyhedra with rational coefficients.
As said previously, $Poly$ is not a complete lattice, and then we cannot define
a Galois connection between $Poly$ and the lattice of the solutions.
Nevertheless, the concretization function  
$\gamma_{poly} : Poly \rightarrow \parties{\mathcal{D}}$ can be defined as the function that
returns the integer points of a given polyhedron:
$$ \gamma_{poly}(\sh{S}) = \textrm{int\_sol}(\sh{S}) $$
Here,  $\textrm{int\_sol}$ stands for the whole set of integer solutions of a set of linear constraints
As $\sh{S}$ is bounded, $\gamma_{poly}(\sh{S})$ is finite.

Without a Galois connection, we do not expect the polyhedral-consistency
proposed in this section to be optimal w.r.t. the abstract domain.
Hence, we only show that the filtering algorithm that computes this
consistency is a sound approximation of the exact filtering operator.

\begin{definition}
Let $\alpha_{box}$ be the following abstraction function\\
$\alpha_{box} : \sh{\mathcal{D}_{bound}} \rightarrow Poly$
such that
$$ \alpha_{box}(([a_1,b_1],\ldots,[a_m,b_m])) = \{a_1 \leq x_1 \leq b_1, \ldots, a_m \leq x_m \leq bm\} $$
and the concretization function $\gamma_{box} :  Poly \rightarrow \sh{\mathcal{D}_{bound}}$:
$$
\gamma_{box}(P) = \left\{ \begin{array}{l}([\lceil min(x_1,P) \rceil, \lfloor
  max(x_1,P) \rfloor],\ldots,[\lceil min(x_m,P) \rceil, \lfloor
  max(x_m,P) \rfloor]) \\\textrm{~~~~~~ if~}\forall i,  \lceil min(x_i,P) \rceil \leq \lfloor
  max(x_i,P) \rfloor \\
  \emptyset \textrm{ otherwise}
\end{array} \right.
$$
\end{definition}
\noindent
where $\lfloor x \rfloor$ (resp. $\lceil x \rceil$) stands for the next smallest (resp. largest) integer of $x$,
and $min(v,P)$ ( resp. $max(v,P)$) computes the smallest (resp. largest) value of $v$ corresponding to a point of $P$.

\noindent
Both $\alpha_{box}$ and $\gamma_{box}$ link the polyhedral abstract domain with the interval abstract domain.
The abstraction function $\alpha_{box}$ maps a set of intervals into a polyhedron by adding two inequalities per variable, while
the concretization function $\gamma_{box}$ maps a polyhedron into a set of intervals by computing first the smallest hypercuboid 
containing the polyhedron and second the greatest hypercuboid with integer bounds.
The behaviour of these two functions is illustrated in Fig. \ref{fig:galois_box}.

\begin{figure} 
  \psfrag{alphabox}{$\alpha_{box}$}
  \psfrag{gammabox}{$\gamma_{box}$}
  \begin{center} 
   \includegraphics[width=11cm,height=5cm]{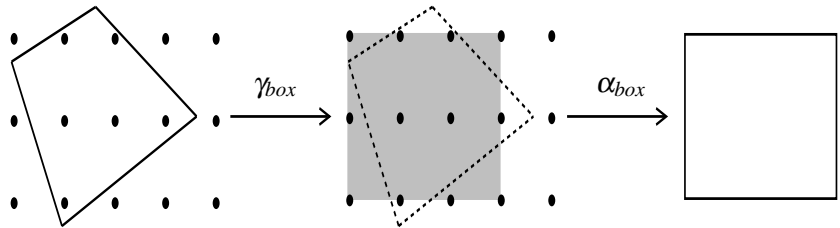} 
  \end{center} 
  \caption{Connection between the Polyhedral and Interval abstract domains} 
  \label{fig:galois_box} 
\end{figure}

\begin{definition}[Polyhedral-consistency]
A domain $\mathcal{D}$ is polyhedral-consistent for a constraint $c$ where $vars(c)=\{x_1,..,x_n\}$ 
iff for each variable $x_i$, $1 \leq i \leq n$ and for each $d_i \in \{inf_{\mathcal{D}} x_i, sup_{\mathcal{D}} x_i\}$
there exist rationals $r_j$ with $inf_{\mathcal{D}} x_j\leq r_j \leq sup_{\mathcal{D}} x_j$, $1 \leq j \leq n, j \neq i$
such that $\sigma = \{x_1 \mapsto r_1, ..,x_n \mapsto r_n\}$ is a (rational) solution of a linear relaxation of $c$.
\end{definition} 
\noindent
The rationale behind this definition is to benefit from efficient polyhedral techniques over the rationals to filter 
the variation domain of variables. Of course, interesting implementations of this filtering consistency should trade 
between efficiency and precision as integer linear constraint solving is costly (NP\_hard problem) even for bounded domains.  
It is worth noticing that the definition depends on the quality of the underlying linear relaxation. On the one hand, a
linear relaxation which over-approximate $c$ by $True$ (the whole search space) is useless while on the other hand 
a linear relaxation which exploits piecewise over-approximations of $c$ is often too costly. 
We give examples of polyhedral-consistency filtering in function of various linear relaxations.

\begin{example}
Consider the following $CS$: $z = x+y, z = x*y$, let $c$ be the second constraint of $CS$: $c = (z = x*y)$
and let $\mathcal{D}$ be $x \in -7 .. 10, y \in -7 .. 10, z \in 3 .. 10$.

\noindent
Note that $\mathcal{D}$ is bound-consistent for all the constraints of $CS$.

\noindent
The simplest linear relaxation that can be considered is the one that ignores non-linear constraints.
In this example, $c$ is over-approximated by $True$  and then $\mathcal{D}$ viewed as  
$x \geq -7, x \leq 10, y \geq -7, x \leq 10, z \geq 3, x \leq 10, z = x+y$ is then
polyhedral-consistant w.r.t. this linear relaxation. Note that this approach can be generalized by
associating a new fresh variable to the non-linear term $x*y$ with a domain computed using the
bounds $x$ and $y$. In this example, this does not help but it could help on other examples.

\noindent
Another linear relaxation consists in building a polyhedron from the ``bounds'' of $x*y$
in $\mathcal{D} = x \in -7 .. 10, y \in -7 .. 10, z \in 3 .. 10$. By considering the 2-dimensional 
polyhedron\\ 
$\{(1,10), (10,1), (-1,-7),(-7,-1)\}$ we get that a linear relaxation of $c$ in
domain $\mathcal{D}$ is\\ 
$11x - 8y +69 \geq 0$\\
$-x -y +11 \geq 0$\\
$-8x + 11y + 69 \geq 0$\\
$x + y + 8 \geq 0$\\  
Filtering with the polyhedral-consistency, we get that
$x \in -2 .. 9, y \in -2 .. 9, z \in 3 .. 10$ where $D(x)$ and $D(y)$ have been pruned.
These results can be easily computed using a Linear Programming tool and truncation operators. For example,
using the clpq library of SICStus Prolog which implements
a simplex over the rationals, the following request permits to compute the max bound of variable $x$:
\begin{verbatim} 
{X >= -7, X =< 10, Y >= -7, Y =< 10, Z >= 3, Z =< 10, Z = X+Y, 
11*X - 8*Y+ 69 >=0, -X - Y + 11 >= 0, -8*X + 11*Y +69 >= 0, 
X + Y + 8 >=0}, sup(X, R).

R = 179/19   % then max bound of x is 9  
\end{verbatim}

Finally, we can automate the computation of linear relaxations of $c$
by considering the following trivial constraints, which are always true for
any $x$ and $y$:
$(x-inf_{\mathcal{D}} x)(y - inf_{\mathcal{D}} y) \geq 0$\\
$(x-sup_{\mathcal{D}} x)(y - inf_{\mathcal{D}} y) \leq 0$\\
$(x-inf_{\mathcal{D}} x)(y - sup_{\mathcal{D}} y) \leq 0$\\
$(x-sup_{\mathcal{D}} x)(y - sup_{\mathcal{D}} y) \geq 0$\\

\noindent
By decomposing these constraints, using the original bounds of $x,y,z$ and 
replacing the quadratic term $x*y$ by $z$, we get:\\
$7x + 7y  +z +49   \geq 0$\\ 
$10x - 7y -z +70   \geq 0$\\
$-7x+10y  -z  +70  \geq 0$\\
$-10x -10y +z  +100 \geq 0$\\
Filtering with the polyhedral-consistency, we get that
$x \in -2 .. 9, y \in -2 .. 9, z \in 3 .. 10$ where $D(x)$ and $D(y)$ have been pruned.
These domains are still bound-consistent but 
another tighter relaxation can be computed with these new bounds:\\
$2X + 2Y  +Z + 4   = 0$\\
$9X -  2Y   -Z +18   = 0$\\
$-2X+  9Y  - Z +18   = 0$\\
$-9X -  9Y + Z  +81 = 0$\\
and then filtering again permits to get that
$x \in 0 .. 8, y \in 0 .. 8, z \in 3 .. 10$.
Here, filtering by bound-consistency leads to prune the domains to:
$x \in 1 .. 8, y \in 1 .. 8, z \in 3 .. 10$.
Then, by iterating these two process, we get the only solution to $CS$ which is:
$x \in 2 .. 2, y \in 2 .. 2, z \in 4 .. 4 $.
This showed how dynamic linear relaxations can be used to solve a non-linear $CS$.
\end{example}

\section{The \textit{w} constraint operator}

In this section, we present the {\it w} constraint operator which captures iterative computations, and how it is processed by a constraint solver. 
The constraint operator has been introduced a long time ago in \cite{GBR98,GBR00} and was further refined using Abstract Interpretation (AI) techniques \cite{DGD07a}. In the following, we recall its interface and semantics and show how fixed point computations can be used to filter 
inconsistant values of the underlying relation. We also explain
how the Polyhedral abstract domain is used to approximate the fixed point computations.

\subsection{\textit{w} as a relation over memory states}
The $w$ operator captures a relation over three memory states that represent the state before, within and after the execution of an iterating statement.
In this paper, we do not specify what a memory state is, or what the iterating statement is, as the approach is generic regarding the content of a memory state and the concrete syntax of the iterator. However, in order to ease the understanding, the reader can consider a memory state to be a mapping between variables of the program to values. More complex examples of memory states in relation with $w$ can be found in \cite{CBG09} and \cite{CG10}.

\noindent
The relation {\it w} is expressed with the following syntax:
$w(\M_1,\M_2,\M_3, Dec, Body)\}$ where $\M_1$ denotes the memory state before execution of the iteration, $\M_2$ denotes the memory state reached at the end of execution of the $Body$, while $\M_3$ denotes the state after execution, $Dec$ is a boolean syntactical expression, and $Body$ is a list of statements. This three-states consideration is inspired by the Static Single Assignment of a program \cite{WZ91}. If the state of $\M$ is irrelevant for a given computation, we simply write $\_$. Note that $Body$ may also contain other iterators, and thus $w$ is meant to be a compositional operator. The semantics of $w$ is the semantics of an iterating statement (i.e., repetitive application of $Body$ over an input state, while $Dec$ is true).

\noindent
We note $w^n = \overbrace{w \circ w \ldots w}^n$ where $\circ$ is the application composition.

\subsection{Background on $w$}
As described in \cite{GBR00},
the operational semantics of $w$ within a constraint solver is expressed as a set of guarded-constraints: $\{(C_1 \longrightarrow C_2)_i\}_{1\leq i \leq n}$. If $C_1$ is entailed by the constraint store then $C_2$ is added to it, and the relation $w$ is solved. If $C_1$ is disentailed, then the guarded-constraint is discarded and no more considered in further analysis. Finally, if none of these (dis-) entailment deductions is possible, the guarded-constraint just suspends in the constraint store. The set of guarded-constraints is considered each time the constraint $w$ awakes in the constraint store, so that it captures the essence of the iteration through rewriting in recursive calls. In addition, substitution of variables must be considered to faithfully represent the constraints in a $w$ relation. $Dec_{\M_3 \leftarrow \M_1}$ simply denotes the constraint $Dec$ where program variables from $\M_3$ have been substituted by the variables from $\M_1$. With these notations, the $w$ relation is expressed as follows:\\
\( w(Dec, \M_1, \M_2, \M_3, Body)   \mbox{  iff  } \)\\ \\
\(\begin{array}{l l l}
    \bullet\ Dec_{\M_3 \leftarrow \M_1} & \longrightarrow   & Body_{\M_3 \leftarrow \M_1} \land w(Dec, \M_2,\M_{new},\M_3, Body_{\M_2 \leftarrow \M_{new}}) \\  
    \bullet\ \neg( Dec_{\M_3 \leftarrow \M_1}) &    \longrightarrow       & \M_3 = \M_1    \\
    \bullet\ \neg(Dec_{\M_3 \leftarrow \M_1} \land Body_{\M_3 \leftarrow \M_1} ) & \longrightarrow &  \neg( Dec_{\M_3 \leftarrow \M_1} ) \land \M_3 = \M_1 \\
    \bullet\ \neg(\neg Dec_{\M_3 \leftarrow \M_1} \land \M_3 = \M_1) &  \longrightarrow     &  Dec_{\M_3 \leftarrow \M_1} \land Body_{\M_3 \leftarrow \M_1} \land w(Dec,\M_2,\M_{new},\M_3, Body_{\M_2 \leftarrow \M_{new}}) \\
  \end{array}
\)\\
\(    \bullet\ join( Dec_{\M_3 \leftarrow \M_1} \land Body_{\M_3 \leftarrow \M_1} \land w(Dec,\M_2,\M_{new},\M_3, Body_{\M_2 \leftarrow \M_{new}}) ,  \neg(Dec_{\M_3 \leftarrow \M_1}) \land \M_3 = \M_1)
\)\\
\noindent
The two former guarded-constraints implement forward analysis, by examining the entailment of $Dec$. Depending on the entailment of $Dec$, a recursive call to a new $w$ is added to the constraint store. The two followings implement backward reasoning by examining the differences between the stores after and before execution of the iteration. Finally,  the last operation, called $join$, is the most tricky one and implements union of stores in case of suspension of the operator. This $join$ operation is realized iff none of the previous guarded-constraints has been solved.
The rest of the Section is devoted to the presentation of this operator, which is implemented as an abstract operation over abstract domains.

\subsection{Concrete fixed point computation}
For a given $w$ operator, let $T$ be the following set:
\begin{displaymath}
T = \{ (\M_i,\M_j)~|~\exists k~|~w^k(\M_i, \M_j,\_, Dec, Body) \}
\end{displaymath}

\noindent
$T$ represents all pairs of memory states that are in relation through the {\it w} statement, but still, not all those pairs can be considered as solutions of the relation, as some pairs can only be reached in temporary states of the execution. For this reason, we introduce the set $Z_w$:

  \begin{eqnarray*}
  Z_w &=& \{(\M_i,\M_j)~|~ (\M_i,\M_j) \in T \land \M_j \in sol(\lnot Dec)\} 
  \end{eqnarray*}

\noindent
  where $sol(C)$ denotes the set of solutions of a constraint $C$.

\noindent
$T$ can be seen as the {\it least fixed point} of:
\begin{eqnarray}
T^{i+1} &= & \{(\M_k,\M_j)~|~(T^i \land w(\M_k, \_, \M_j,Dec, Body))\} \cup T^i \label{eqn:concrete1}\\
T^0 &= &\{(\M_1,\M_1)\}  \label{eqn:concrete2} 
\end{eqnarray}
and $Z_w$ can be computed by filtering the pairs of the fixed point.

\noindent
For instance, considering $\M_1 = x \mapsto 0 \vee x \mapsto 1 \vee x \mapsto 2 \vee x \mapsto 3$ and $w(\M_1,\M_2, \M_3, x < 2, x=x+1)$, and
using the notation $(0,0)$ for denotating $(x \mapsto 0, x \mapsto 0)$,
the fix point computation is as follows: 

\begin{eqnarray*}
T^0 &=& \{(0,0),(1,1),(2,2),(3,3)\} \\ 
T^1 &=& \{(0,1),(1,2)\} \cup T^0
    = \{(0,0),(0,1),(1,1),(1,2),(2,2),(3,3)\}\\
T^2 &=& \{(0,1),(0,2),(1,2)\} \cup T^1
    = \{(0,0),(0,1),(0,2),(1,1),(1,2),(2,2),(3,3)\}\\
T^3 &=& T^2     
\end{eqnarray*}

\noindent
Consequently, the solutions set $Z_w$ of $w(\M_1, \M_2, \M_3, x < 2, x=x+1)$ is:
\begin{eqnarray*}
  Z_w &=& \{(a,b)~|~(a,b) \in T^3 \land (x \mapsto b) \in sol(x \geq 2)\}  \\
   &=& \{(0,2),(1,2),(2,2),(3,3)\}
\end{eqnarray*}

\noindent
Computing $Z_w$ is undecidable in general as there is no termination guarantee of the iterating process. This is the reason why this computation is usually abstracted using abstract domain computation.

\subsection{Abstracting the fixed point computation}

Implementing the $join$ operator mentionned above can be done by 
abstracting the computation of the fixed point within the Polyhedral abstract domain.
Let $P^\sharp$ be a conjunction of linear restraints, the intersection of which defines a convex polyhedron, that over-approximates the set $T$.
Hence, we can compute $P^\sharp$ as the least fixed point of:
\begin{eqnarray}
P^{i+1} &=& \{(\M_k,\M_j)~|~(P^i \land \alpha_{poly}(~w(\M_k,\M_j,\_,Dec,Body))) \sqcup P^i \label{eqn1}\\
P^0 &=& \{(\alpha_{poly}((\M_1,\M_1))\label{eqn2}
\end{eqnarray}
Compared to eq.~\ref{eqn:concrete1} and~\ref{eqn:concrete2}, the
computation is realized in the abstract domain using $\alpha_{poly}$ the abstraction function of the Polyhedral abstract domain.



\noindent
Let $Z^\sharp_w$ be the approximation of the set of solutions of \textit{w}, obtained by application of $\alpha_{poly}$:
\begin{eqnarray*}
  Z^\sharp_w &=& \{(\M_i,\M_j)~|~ (\M_i,\M_j) \in P^\sharp \land \M_j \in \alpha_{poly}(sol(\lnot Dec))\} 
  \end{eqnarray*}

\noindent
Looking at the above example where $\M$ is just composed of the mapping of $x \mapsto v$, 
it is worth introducing different representations of the stores as we progress in the fixed point computation. 
When $P^i$ is computed over $x_k$ and establishes a relation in between stores $\M_k$ and $\M_j$ that contains $x_j$, we note: $P^i(x_k,x_j)$. If $P^i$ is then
considered over $y_k,y_j$, then we will simply write $P^i(y_k,y_j)$ and apply variable substitution. 

\noindent
With these notations, we have the following computation: 
 \begin{eqnarray*}
 P^0(x_{in},x_{out}) &=&  x_{in} \geq 0 \land x_{in} \leq 3 \land x_{in} = x_{out} \\
 P^1(x_{in},x_{out}) &=& (P^0(x_{in},x_0) \land x_0 \leq 1 \land x_{out} = x_0 + 1)_{x_{in},x_{out}}
       \sqcup~P^0(x_{in},x_{out}) \\
    &=& (x_{in} \geq 0 \land x_{in} \leq 1 \land x_{out} = x_{in} + 1) \sqcup P^0(x_{in},x_{out}) \\
    &=& x_{in} \geq 0 \land x_{in} \leq 3  \land x_{out} \leq x_{in} + 1 \land  x_{out} \geq x_{in} \\
 P^2(x_{in},x_{out})  &=&  (P^1(x_{in},x_1) \land x_1 \leq 1 \land x_{out} = x_1 + 1)_{x_{in},x_{out}}
       \sqcup~P^1(x_{in},x_{out}) \\
    &=& (x_{in} \geq 0  \land x_{in} \leq  3 \land x_{in} \leq  x_{out} - 1) \sqcup P^1(x_{in},x_{out}) \\
    &=&  x_{in} \geq 0 \land x_{in} \leq 3  \land x_{out} \leq x_{in} + 2 \land  x_{out} \geq x_{in} \land x_{out} \leq 4\\
 P^3(x_{in},x_{out}) &=&  (P_2(x_{in},x_2) \land x_2 \leq 1 \land x_{out} = x_2 + 1)_{x_{in},x_{out}}
    \sqcup~P^2(x_{in},x_{out}) \\
 &=& (x_{in} \geq 0  \land x_{in} \leq  3 \land x_{in} \leq  x_{out} - 1) \sqcup P^2(x_{in},x_{out}) \\
 &=& P^2(x_{in},x_{out})
 \end{eqnarray*}

\noindent
Fig.~\ref{approx} illustrates the difference between the abstract fixed point and the approximate fixed point.
Points in the figure correspond to the elements of $T^3$, while the grey zone represents the convex polyhedron defined by $P^3$. 

\begin{figure}
\begin{center}
  \includegraphics[width=3cm,height=3cm]{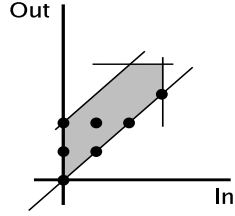}
\end{center}
\caption{Exact and approximated fixed point}
\label{approx}
\end{figure}

\noindent
An approximation of the solutions of $w(\M_1, \M_2, \M_3, x < 2, x=x+1)$ is given by: 
\begin{eqnarray*}
Q &=& P^3(x_1,x_3) \land x_3 \geq 2 \\
  &=& x_3 \geq 2 \land x_3 \leq 4 \land x_1 \leq x_3 \land x_1 \leq 3 \land  x_1 \geq x_3 - 2
\end{eqnarray*}

\noindent
On the Polyhedral domain, convergence of the fixed point computation over $w(\M_1, \M_2, \M_3, x < 2, x=x+1)$ can be enforced by using widening techniques.
The computation of $P^{k+1}$ is modified in order to use a widening operator $\nabla$ \cite{CH78}. Thus, we have: 

\begin{eqnarray*}
P^{k+1} &=& P^k(Init,Out) \nabla (P^k \land  \alpha_{poly}(w(\M_1,\M_2, \M_3, Dec,Body))) 
\end{eqnarray*}

\noindent
A concrete algorithm for computing this approximation is given in \cite{DGD07a}, which permits to build implementation of $w$ in a constraint solver. As rooted in the Abstract Interpretation domain, the relation $w$ inherits from some of its fundamental correctness results, i.e., soundness and termination.  However, it is worth pinpointing some differences.

\noindent
Usually, a convex abstract polyhedron denotes the set of linear relations that hold
over variables at a given point of a sequential program under analysis. 
As the goal here is to correctly approximate the set of solutions of a \textit{w} relation, the polyhedron describes
relations between input and output values and, thus, they
involve more variables in the equations. In Abstract Interpretation, the
analysis can be performed only once, whereas, in the case of the $w$ relation, the $join$ operation is launched everytime the relation is awaked without being succesfull in solving one of the guarded-constraint.
As a consequence, we found out that it was not reasonable to use standard libraries to compute over polyhedra, such as PPL~\cite{BRZ02}, 
because they use a {\it dual representation for Polyhedra}, which is a source of exponential time
computations for the conversion.  

\subsection{Illustrative example}

Looking at an iterative computation over unbounded domains as a relation captured by a $w$ constraint operator is 
interesting for adressing Constraint-Based Reacheability problems. On the one hand, the suspension mechanism offered by constraint reasoning allows us 
to cope with the approximation problem, i.e., the set of states that is considered is determined by the informations existing 
in the constraint store, which makes the reasoning more accurate w.r.t. the property to be demonstrated. On the other hand, adding abstract domain 
computations to the $w$ relation allows us to increase the level of deductions that can be achieved at each awakening of the $w$ constraint operator. 
To illustrate this remark, consider the following $C$ program:
\begin{verbatim}
f(  int i, ...  )  {
a.    j = 100;
b.      while( i > 0)
c.        { j=j+1 ; i=i-1 ;}
d.      ... 
e.	if( j > 500)
f.	    ...
\end{verbatim}
\noindent
A typical reachability problem is to find out a value of {\tt i} such that statement {\tt f.} is executed. Existing approaches for solving this reachability problem consider a path passing through {\tt f.}, e.g., {\tt a-b-d-e-f}, and try to solve the {\it path condition} attached to this path. In this case, it means extracting constraint $j_1 = 100 \land i_1 \leq 0 \land j_1 > 500$ and solving it to show that the constraint system is unsatisfiable, i.e., the corresponding path is infeasible. Then, these approaches backtrack to select another path (e.g., {\tt a-b-c-b-d-e-f} with path condition 
$j_1 = 100 \land i_1 > 0 \land j_2 = j_1+1 \land i_2 = i_1-1 \land i_2 \leq 0 \land j_2 > 500$) and repeat the process again, until a satisfiable path condition is found. This example is pathologic for these approaches, as only the paths that iterate more than $400$ times in the loop will reach statement {\tt f.}. Hopefully, using the constraint operator $w(\M_1, \M_2, \M_3, i > 0, j=j+1 \land i=i-1)$ permits us to unrool dynamically $400$ times the loop without backtracking. The relational analysis performed on the Polyhedral abstract domain by the $w$ operator determines that $j_{out} - i_{in} = 100$ whatever be the number of loop unrollings. Here, combining precise constraint reasoning in the concrete domain, with constraint extrapolation through abstract domain computations, offers us an efficient way of solving reachability problems on infinite-state systems.



\section{Conclusions}

In this paper, we have presented Constraint-Based Reachability as a process to combine constraint reasoning and abstraction techniques for solving reachability problems in infinite-state systems. The contribution is two-fold: first, we have revisited constraint consistency-filtering techniques by the prism of abstract domain computations~; second, we explained how to introduce abstract domain computation within the $w$ constraint operator reasoning. We have illustrated these notions with several examples in order to ease the understanding of the reader.

This appraoch has been implemented and tested on several problems, including real-world programs \cite{Got09,Got12}. The goal is now to broader the scope of these techniques that combine constraint reasoning and abstraction techniques, to adress fundamental problems such as reachability in infinite-state systems.

\section*{Acknowledgements}
We are indebted to Bernard Botella and Mireille Ducass\'e for fruitful discussions on earlier versions of this work.

\section*{Appendix}
This appendix contains the proofs of some of the results stated in the paper.
\bigskip

\begin{theorem}
\label{theo:arc}
  Let $p$ be a filtering operator associated with constraint $c_i$, then $p$ computes domain-consistency iff
  $p = \sh{f_{i\_arc}}$. 
\end{theorem}

\begin{proof}
  ($\Leftarrow$) Let $S_1 = (f_i \circ \gamma_{arc})(S)$. From the definitions of
  $f_i$ and $\gamma$, we get that $S_1$ is the solution set of constraint $c_i$,
  given the initial domains $S$ (we write $S_1 = sol(c_i,S)$).  Hence, $S' = \alpha_{arc}(S_1) =
  (A_1, \ldots, A_m)$ with \\$A_k = \{x_k~|~x \in sol(c_i,S)\}$.
  So, $c_i$ computes domain-consistency.\\
  ($\Rightarrow$) Let $p$ be a domain-consistency filtering operator. 
  Suppose that there exists $S$ such that 
  $p(S) = (A_1,\ldots,A_m)$ be strictly greater than
  $\sh{f_{i\_arc}}(S) = (B_1, \ldots, B_m)$.  Then, there exists at least 
  one $k$ such as $A_k \supsetneq B_k$. Hence, there exists an
  element $x_k$ of $A_k$ that does not belong to any solution of constraint $c_i$. 
  Hence, $p$ cannot computes domain-consistency which is contradictory with the hypothesis.
  On the other side, $p$ cannot be smaller than
  $\sh{f_{i\_arc}}$ as it means that the filtering operator removes solutions.
  Hence, if $p$ computes domain-consistency then $p = \sh{f_{i\_arc}}$.
 $\Box$
\end{proof}

\begin{theorem}
 If $p$ is a filtering operator associated to constraint $c_i$, then  $p$ computes 
 bound-consistency iff $p = \sh{f_{i\_bound}}$. 
\end{theorem}

\begin{proof}
  ($\Leftarrow$) From theorem~\ref{theo:arc}, given initial intervals $I$, 
  the domains $\sh{f_{i\_arc}} \circ \gamma_{inter}(I)$ are domain-consistent for constraint $c_i$.
  Applying function $\alpha_{inter}$ is similar to the process that keeps extremal values of each element of 
  $\sh{f_{i\_arc}} \circ \gamma_{inter}(I)$. Hence, the resulting intervals satisfy the 
  bound-consistency property.\\
  ($\Rightarrow$) (similar to the proof of theorem~\ref{theo:arc}) 
  If the filtering operator $p$ is greater than $\sh{f_{i\_bound}}$, 
  then the computed intervals contain at least one bound that is not part
  of a solution of $c_i$, violating so the bound-consistency property.
  On the contrary, by supposing that $p$ is smaller than
  $\sh{f_{i\_bound}}$ then solutions are lost and $p$ is no more a filtering operator.
  Hence, if $p$ is a filtering operator guaranteeing bound-consistency
  then $p = \sh{f_{i\_bound}}$. $\Box$
\end{proof}

\bibliographystyle{eptcs}
\bibliography{../../../these,../../../AG_publis}
\end{document}